\documentclass[12pt]{article}
\usepackage{indentfirst}
\usepackage{arabtex}
\usepackage{amssymb}
\usepackage{amsmath,amsthm}
\usepackage{color, xcolor}
\usepackage[ruled,vlined]{algorithm2e}
\usepackage{graphicx}
\usepackage{url}
\usepackage{authblk}
\usepackage{pgfplots}
\usepackage[inline, shortlabels]{enumitem}

\usepackage{tikz-qtree}
\usetikzlibrary{calc}
\usetikzlibrary{shapes,arrows,chains}
\usetikzlibrary{automata,positioning}
\tikzset{edge from parent/.append style={<-}}
\tikzstyle{sn}=[circle, draw, align=center, fill=white!50, inner sep=0pt, minimum size=0.6cm]

\tikzstyle{kn}=[circle, draw, align=center, fill=white!50, inner sep=0pt,
minimum size=1.6cm]

\tikzstyle{fsn}=[circle, draw, fill=blue!30, inner sep=0pt, minimum size=0.6cm]

\newtheorem{theorem}{Theorem}

\newtheorem{lemma}{Lemma}

\newtheorem{observation}{Observation}

\pgfplotsset{votegas/.style={%
  ytick={1000000, 2000000, 3000000, 4000000, 5000000, 6000000,
  7000000, 8000000},
  ylabel= {Gas Cost},
  xlabel = {Length of delegate chain},
  ymin=0,
  ymax=9000000,
  xmax = 3200,
  xmin = 0,
  scale only axis
}
}
\begin{document}
	\pagestyle{empty}

	\title{Implement Liquid Democracy on Ethereum: A Fast Algorithm for Realtime Self-tally Voting System}
\author[1]{Xuepeng Fan \thanks{xuepeng.fan@asresearch.io}}
\author[2]{Peng Li \thanks{pengli@u-aizu.ac.jp}}
\author[1]{Yulong Zeng \thanks{yulong.zeng@asresearch.io}}
\author[2]{Xiaoping Zhou \thanks{d8211110@u-aizu.ac.jp}}
\affil[1]{Autonomous System Research}
\affil[2]{The University of Aizu}

\maketitle
\begin{abstract}
	We study the liquid democracy problem, where each voter can either directly vote to a candidate or delegate his voting power to a proxy. We consider the implementation of liquid democracy on the blockchain through Ethereum smart contract and to be compatible with the realtime self-tallying property, where the contract itself can record  ballots and update voting status upon receiving each voting massage.  A challenge comes due to the gas fee limitation of Ethereum mainnet, that the number of instruction for processing a voting massage can not exceed a certain amount, which restrict the application scenario with respect to algorithms whose time complexity is linear to the number of voters. We propose a fast algorithm to overcome the challenge, such that i) shifts the on-chain initialization to off-chain and ii) the on-chain complexity for processing each voting massage is $O(\log n)$, where $n$ is the number of voters. 
\end{abstract}

\section{Introduction}
Democracy has always been a widely concerned topic. Voting, as a primary method for realizing democracy in modern society, is more and more common in practice, with applications ranging from the presidential election to community governance \cite{gerber1998primary,kohno2004analysis}. Meanwhile, various issues emerge due to the current voting system, with low participation \cite{birch2016full,kousser2007does} and black-box operation\cite{bannet2004hack,harris2004black}\footnote{According to \cite{harris2004black} Black Box Voting is: ``Any voting system in which the mechanism for recording and/or tabulating the vote is hidden from the voter, and/or the mechanism lacks a tangible record of the vote cast"} to be two of the most significant. Now, a group of technologist are looking for a new approach to reform the voting system, bring all people with voting rights closer to their representatives and holding elections in a public, verifiable way. In other words, voters should exercise their rights in every related voting, e.g. a policy issue or a piece of new legislation, in which the voting status should be publicly displayed and traceable. However, usually people do not have time for full participation, or they are not expert with respect to the area involved by the voting proposal, which resulted in a large number of voting powers not actually being exercised.

The concept liquid democracy (also known as proxy voting, delegate voting) is proposed to handle the participation issue, in which the core idea is that, each voter (also called delegator) can select a personal representative who has the authority to be a proxy (also called delegate) for his vote.  Those delegates can further proxy their votes to other people, creating a directed network graph (called delegate graph). Whenever a delegate votes to a candidate (or a proposal, a policy), by default all his delegators' (including multi-level) voting powers are also cast to that candidate. If the delegator dislikes the way in which the proxy voted, they can either vote themselves, which dilutes the proxy's power, or pick another delegate for the next vote. Applying liquid democracy system significantly reduces time costs of voters and increases voting participation, which has been studied over a long period of time.

The early proposal about liquid democracy can be traced back to 1884 by Lewis Carroll \cite{carroll1884principles}, and followed by a number of economists \cite{william1912primary,miller1969program,paulin2014through}. Nowadays, many companies/parties also implement practical applications of liquid democracy, such as Google votes \cite{hardt2015google}, Pirate Parties (software: liquid feedback) \cite{behrens2014principles}, etc. However, they are still centralized organization, where black-box operation and statistical error are inevitable.

The introduction of blockchain \cite{nakamoto2008bitcoin} and on-chain smart contract technology \cite{wood2014ethereum} satisfies the openess and transparency requirement of voting system. Generally speaking, a blockchain is a decentralized and immutable public ledger ensured by cryptography and P2P networking, storing data including transaction information which is observerble by any user. The smart contract is a pre-designed instruction set for storing and operating on-chain data, led by Ethereum. The source code of the voting system can be deployed on the Ethereum mainnet through smart contract and invoked through on-chain transactions. 
The decentralization of blockchain gurantees that the voting system are impartially executed without any need of trusted third party, eradicating black-box operations.

 A fundamental requirement of on-chain voting systems is realtime self-tallying, which states that, for each incoming voting message,  the contract itself can record ballots\footnote{We uses ballot to denote the number votes that a candidate receives.} and update the voting status (and display it). The self-tallying property skips the trouble to download the whole Ethereum data (and use it to off-chain compute the results, mainly for those who do not participate the voting), or to collect majority’s signatures confirming a specific voting status.\footnote{With the self-tallying property, users can visits the variables of a smart contract by simply send RPC requests to Ethereum full nodes. Otherwise, without the support of any centralized parties like {\em etherscan.io}, users need to download the full Ethereum in order to obtain historical visiting data of the smart contract (current size more than 3.4TB).}

However, the main challenge for realizing on-chain liquid democracy is the limitation of gas fee on Ethereum mainnet, which is remain open\footnote{https://forum.aragon.org/t/open-challenges-for-on-chain-liquid-democracy/161}: executing a single instruction of smart contracts consumes a certain amount of gas, called gas fee, with the average about 10 thousand\footnote{According to https://github.com/ethereum/EIPs/blob/master/EIPS/eip-150.md, one storage\_modify instruction costs 5000 gas, and one storage\_add instruction costs 20000 gas.}.  Whereas Ethereum has a parameter \textbf{block\_gas\_limit}, usually about 10 million, which determines the total gas that can be consumed within a block. That is, the total gas fee for invoking a smart contract can not exceed block\_gas\_limit, because the corresponding transaction can only be included in one block, which means that number of instructions can not exceed 10 million/10 thousand = 1000, otherwise the transaction will be refused. 

 The difficulty lies in the computation of two pieces of critical information upon receiving a voting massage: i) the actual voting power that the voter exercises, which would be reduced since some of his delegators may already cast a direct vote, and ii) the change other candidates' ballots, since the voter's direct vote also reduces the actual voting power of his delegate. Naive algorithms usually compute the two values by traversing through the delegate graph, of which the on-chain time complexity is $O(n)$ for processing each voting message, where $n$ is the number of voters. Especially, when the delegate graph is chain-like, they are undesirable due to the limitation of gas fee, since application scenarios are limited to less than one thousand voters (usually votings with millions of voters are required)

Some discussions try to solve the challenge by add restrictions to the delegate graph, i.e., only allow the delegate graph with the max-depth less than 100. This modification essentially limits voters' behavior, which is lack of user friendliness and does not catch the core of democracy. Moreover, this method is vulnerable under attack: suppose Bob wants to prevent Alice from delegating to anyone, he can create 98 accounts and form a delegate chain, with the top node of the chain delegating to Alice. (Creating new accounts is zero-cost in Ethereum) Anytime Alice delegates to a node, a 100-depth chain is generated, thus the delegation will be refused. Other solutions abound but all unreasonable when meet with the gas fee limitation. 

In this paper we propose an algorithm that reduces the on-chain time complexity to $O(\log n)$ for processing each voting information, which essentially solves the on-chain liquid democracy problem. Our algorithm does not add any restriction to voters: any voters can delegate arbitrary. Our algorithm's off-chain time complexity is also acceptable, only $O(n)$. 

Our algorithm mainly depends on two techniques, the Merkel tree, an on-chain storage method, and the interval tree, a data structure. Our algorithm solves the liquid democracy  problem with the following aspect: 
\begin{itemize}
\item At the beginning of a voting, each voter obtains the delegate graph by snapshotting the current height of Ethereum, then executes a $O(n)$ off-chain initialization to get his initialization data.  
\item  Each voter is not allowed to change his delegate within the period of a voting, but he can directly vote to a candidate by send a voting massage, attached with his initiation data. The Merkel tree method checks the correctness of the initiation data with $O(\log n)$ time complexity. 
\item Upon receiving a voting message, our algorithm requires $O(\log n)$ time complexity for updating/displaying the voting status and storage, through the interval tree structure. 
\end{itemize}

Our on-chain algorithm excluding the Merkel tree part also enhances the off-chain liquid democracy: if each voter votes once, then the time complexity is $O(n^2)$ for traversal algorithms, while our algorithm is $O(n\log n)$. We focus on the on-chain situation in the following of this paper. 

\subsection{Related Work}
The {\em Liquid Democracy Journal}\footnote{https://liquid-democracy-journal.org/} collects many valuable literature about the liquid democracy problem, which begins at 2014 and almost information about latest progress can be found there. Blum and Zuber\cite{blum2016liquid} give an overview liquid democracy, include the concepts, the history and problems. Recently, a few technical papers also are interested in liquid democracy. Anson et al. \cite{kahng2018liquid} analyze the problem that whether there exists a delegate voting that outperforms direct voting, for the situation where there are a correct candidate and an incorrect candidate. Brill and Talmon \cite{brill2018pairwise} study the case where a voter can delegate to several proxies and specify a partial order. They propose a way to overcome the complications of individual rational. Christoff and Grossi \cite{christoff2017binary} analyze liquid democracy within the theory of binary aggregation, and consider the issues of individual rational and delegate cycle. 

 Recently, a series of literature studies the implementation of on-chain voting system \cite{hanifatunnisa2017blockchain,hjalmarsson2018blockchain}, some of which also refer to liquid democracy but not consider the self-tally requirement. The introduction of self-tally can be found in \cite{kiayias2002self}, which states that the property of self-tally and perfect ballot secrecy can not be satisfied simultaneously. Thus in this paper the privacy is compromised in favor of realtime self-tallying. Yang et al.\cite{yang2018decentralized} introduce a self-tallying voting system by Ethereum smart contract, but do not consider the liquid democracy scenario. McCorry et al. \cite{mccorry2017smart} also implement a distributed and self-tally electronic voting scheme using the Ethernet blockchain, while the core is to maximize the protection of voter privacy. 

To the best of our knowledge, our paper is the first $O(\log n)$ algorithm solving the on-chain liquid democracy problem, while the algorithms in Google vote and liquid feedback work in following ways: Google vote's algorithm mainly bases on the work of Schulze's \cite{schulze2011new}, which is a $m^3$ method for electing a winner, where $m$ is the number of candidates. They also demonstrate that the system can implement liquid democracy on a social network in a scalable manner with a gradual learning curve. The basics of Liquid Feedback's algorithm comes from Harmonic Weight\footnote{https://en.wikipedia.org/wiki/Harmonic\_mean}, Proportional Runoff\footnote{http://www.magnetkern.de/prop-runoff/prop-runoff.html} and Schulze method, whose proposes are to determine the weights of candidates. Though both  algorithms can be applied to liquid democracy, the self-tallying requirement and gas fee limitation are not taken into consideration.
\section{Preliminary}
\subsection{Problem Description}
Suppose there are $n$ voters, indexed by numbers $1,2,...,n$, and $m$ candidates, indexed by capital letters $A,B,C,...$. We separate the liquid democracy  problem into three periods:
\begin{itemize}
	\item \textbf{Spare period} During spare period, no voting  is hold. Each voter can arbitrarily delegate, undelegate and change delegate, by sending a massage (transaction) to the blockchain, which is stored in the delegate contract. Each voter is allowed to appoint at most one delegate. 
	\item \textbf{Prepare period}
	In prepare period, a specific voting is to be hold. The holder needs to deploy the voting contract and the delegate graph and each voters voting powers need to be constructed. In the following example we regard the delegate graph as input,  while in the next section we will show how the voting powers are determined and how to handle the case where there is a cycle in the delegate graph. 
	\item \textbf{Voting period} After the voting begins, each voter can directly vote to a candidate by sending a voting message, with all his delegators' voting powers also cast to that candidate (which may also reduce the vote of the candidate that his delegate votes). The on-chain voting status are updated for each voting message and need to be displayed. For convenience, we assume that each voter can only vote once during each voting activity, while our algorithm also fits for the case where each voter can change his vote. A voter's delegate is not allowed to change during voting period.
\end{itemize}
It is notable that, although our algorithm does not allow voters to change their delegates during the voting period, they can accomplish the same thing by casting a direct vote.  Changing delegate to a voter that has voted is equivalent to casting a direct vote and changing delegate to a voter that has not voted is rare in practice. Actually, a voter wants to change his delegate during the voting period usually because  he dislikes the way his delegate votes, which means that he has a better candidate in mind and is better off to vote by himself, and he has the time to do so.

The following example abstracts how voting status changes upon receiving voting massages during voting  period.
 Suppose a (direct, no-cycle) delegate graph $G=(V,E)$ is given, where each node represents a voter, and a direct edge $(u,v)$ represents that voter $v$ delegates to voter $u$. We will use terms ``voter" and ``node" interchangeably in the following of this paper.  Since by assumption there is no
cycle in $G$, thus $G$ is a forest (multiple trees). For convenience, we add a virtual node (indexed 0) that
is pointed by the root of each connected branch. So $G$ is transferred to tree $T$, as
Figure~\ref{fig:1}. That is, there are 12 voters. We further assume that each voter's voting power equals to his index. 

\begin{itemize}
\item At the beginning, nobody
votes. When voter 1 votes for candidate $A$ (as the first voter), $A$ obtains
$1+2+...+12=78$ votes.

\item After voter 1 votes, suppose voter 5 (as the second
voter) votes for candidate $B$. Then $B$ obtains $6+5=11$ votes. $A$'s vote
decreases by 11, turning into 67.

\item Further, voter 3 (as the third voter) votes for candidate $C$, then $C$
obtains $3+4+7+8=22$ votes, $A$'s vote becomes 45, and $B$'s vote is still 11.
\end{itemize}

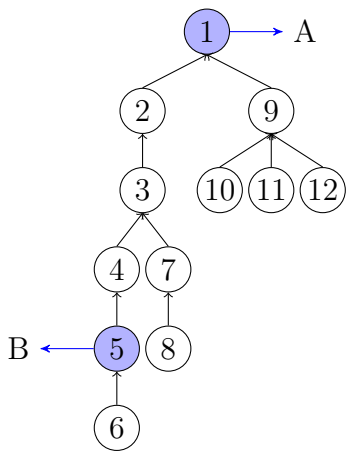
\begin{figure*}
  \centering
  \begin{tikzpicture}
\Tree
[.\node[fsn](n1){1};
  [.\node[sn]{2};
   [.\node[sn]{3};
    [.\node[sn]{4};
     [.\node[fsn](n5){5};
      [.\node[sn]{6};]]]
    [.\node[sn]{7};
     [.\node[sn](n8){8};]]] ]
  [.\node[sn]{9};
   [.\node[sn]{10};]
   [.\node[sn]{11};]
   [.\node[sn]{12};]]
]
\node (a) at ($(n1.east) + (1, 0)$) {A};
\node (b) at ($(n5.west) + (-1, 0)$) {B};
\draw [->, >=stealth', color=blue] (n1) -- (a);
\draw [->, >=stealth', color=blue] (n5) -- (b);

\end{tikzpicture}
  	\caption{Tree $T$. We ignore the virtual node with index 0 here.}
  		\label{fig:1}
\end{figure*}
\textbf{Goal}:
For each voting massage, display the votes of all candidates. (Suppose $m<100$)

\begin{tabular}{|c|c|}
input & output \\
1 A			&		A 78 B 0 C 0
\\
5 B			&		A 67 B 11 C 0
\\
3 C			&		A 45 B 11 C 22
\end{tabular}
\subsection{Blockchain and Smart Contract}
The smart contract of Ethereum supports Turing-complete programming language, which is deployed on the blockchain \cite{bocek2018smart}. Users invoke a smart contract by sending a transaction to the smart contract's address, which contains additional information including the gas fee of the transaction and other incoming parameters, which would further be included in a block. As shown in section 1, the gas fee of a valid transaction are limited by the fixed parameter block\_gas\_limit, so that the number of instructions of the smart contract are also bounded. That is the so-called "on-chain" time complexity. In this paper, we use ``voting massage" to represent the transaction that invokes the voting contract, whose on-chain time complexity are required to be sub-linear to the number of voters.  Ethereum clients obtain latest on-chain data through P2P network and implement smart contracts locally through the Etheruem virtual machine (EVM). Thus, the space limitation of a smart contract only depends on Ethereum nodes' (clients) local storage, which is not an issue in this paper.

It is important to distinguish on-chain smart contract and (open source) cloud computation. The later is still realized by centralized servers, which can not guarantee the codes are correctly executed, while in Ethereum, there is no "center" for executing smart contract: they are executed by every Ethereum node, which is reliable as long as the majority of nodes are honest.

\subsection{Merkel Tree}
The Merkel Tree is a common used method for store and verifying on-chain data. One of the key tools is the hash function, $h()$, which is (cryptographically) hard to find collisions and for inverse computation (In Ethereum, SHA3-256 is used). The Merkel tree is a full binary tree, where each leaf node stores the hash value of data to be stored (see Figure \ref{fig:merkel}). The value of an intermediate node is the hash value of the combination of its two children. 
\begin{figure*}
	\centering
	\label{fig:merkel}
	\includegraphics[width=1.1\textwidth]{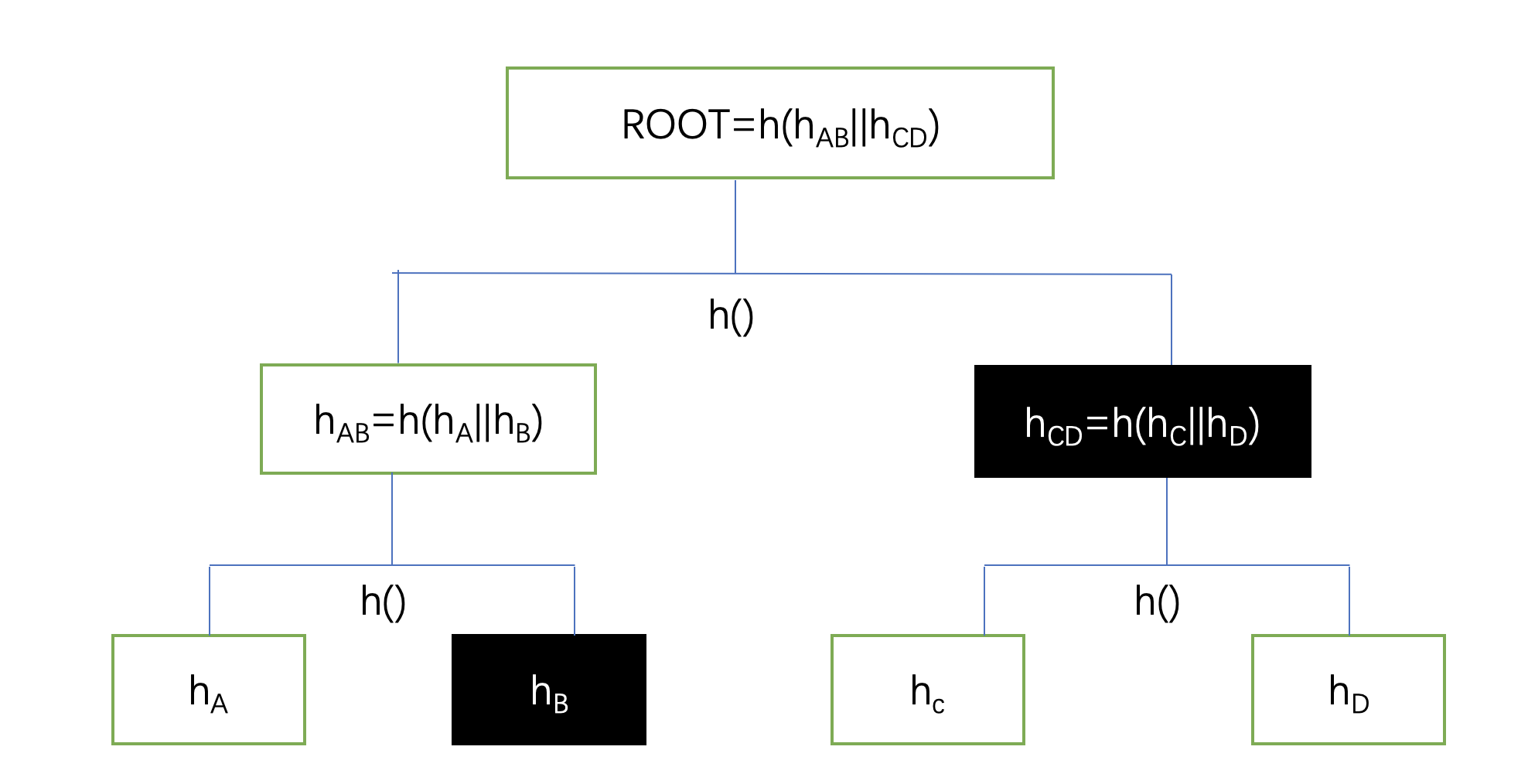}
	\caption{Merkel tree, where $h_{A/B/C/D}$ refers to the hash value of data $A/B/C/D$. The black nodes represent the Merkel path of data $A$.}
\end{figure*}

The use of Merkel tree is that, the blockchain only need to store the root of the Merkel tree. In order to proof that a data (say data $A$ in the Figure \ref{fig:merkel})belongs to the Merkel tree, $A$ along with its Merkel path (also called Merkel Proof) are required:

\textbf{Merkel path}, which is defined to be a sequence of nodes in the Merkel tree that corresponds to brother of each node on path from  $A$'s leaf node to the root. For example, data $A$'s Merkel path is $(h_B,h_{CD})$. A leaf node together with its correct Merkel path can recover the root of the Merkel tree (called Merkel root). The length of the Merkel path and the time complexity for recovering the Merkel root are all logarithmic to the number of leaf nodes of the Merkel tree. The one-wayness of the hash function guarantees that it is hard to construct a correct Merkel path for any data that does not belong to any leaf nodes of the Merkel tree.
\subsection{Interval Tree}
Interval tree is also a binary tree, where each node represents a interval and the interval of a parent node is uniformly distributed to its two child nodes, until the interval becomes a singular, to be a leaf node. 
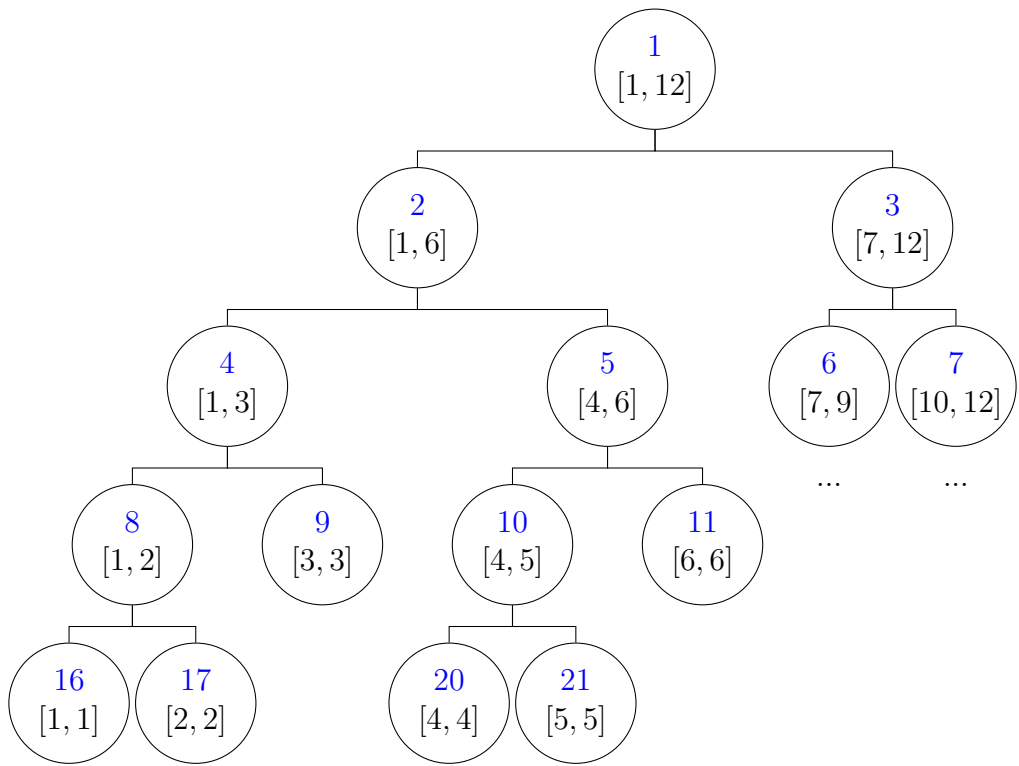
\begin{figure}
	\centering
	\begin{tikzpicture}[level distance=60pt]
\tikzset{edge from parent/.style=
{draw,
edge from parent path={(\tikzparentnode.south)
-- +(0,-8pt)
-| (\tikzchildnode)}}}
\Tree
[.\node[kn](n1){{\color{blue}1}\\$[1,12]$};
  [.\node[kn]{{\color{blue}2}\\$[1,6]$};
   [.\node[kn]{{\color{blue}4}\\$[1,3]$};
    [.\node[kn]{{\color{blue}8}\\$[1,2]$};
     [.\node[kn]{{\color{blue}16}\\$[1,1]$};]
     [.\node[kn]{{\color{blue}17}\\$[2,2]$};]
    ]
    [.\node[kn]{{\color{blue}9}\\$[3,3]$};]
   ]
   [.\node[kn]{{\color{blue}5}\\$[4,6]$};
    [.\node[kn]{{\color{blue}10}\\$[4,5]$};
     [.\node[kn]{{\color{blue}20}\\$[4,4]$};]
     [.\node[kn]{{\color{blue}21}\\$[5,5]$};]
    ]
    [.\node[kn]{{\color{blue}11}\\$[6,6]$};]
   ]
  ]
  [.\node[kn]{{\color{blue}3}\\$[7,12]$};
   [.\node[kn](n6){{\color{blue}6}\\$[7,9]$};]
   [.\node[kn](n7){{\color{blue}7}\\$[10,12]$};]]
]

\node (a) at ($(n6.south) + (0, -0.5)$) {...};
\node (b) at ($(n7.south) + (0, -0.5)$) {...};
\end{tikzpicture}
	\caption{The interval tree of all nodes}
				\label{fig:interval}
\end{figure}
See Figure \ref{fig:interval} for the interval tree of the 12 nodes in Figure \ref{fig:1}. The root are indexed 1 and for node $k$, its child nodes are indexed $2k,2k+1$ respectively. Note that, although some indexes, e.g. 18,19 in Figure \ref{fig:merkel}, does not exist, the space is still used. 

Interval tree is usually used for maintaining an array where the operations are aiming at a successive interval. Given the interval to be updated, the execution begins at the root node, then the interval are separated to one or two sub-intervals. which are recursively executed at the child notes, guaranteeing that the sub-intervals to be update belongs to interval of current node. The recursion ends when the interval to be update equals to the interval of current node. For example, when interval 
[4,8] is to be update, beginning at the root, it separates to [4,6] and [7,8], which are executed at note 2 and 3 respectively. The recursion ends at node 5 and node 12 (which are omitted in Figure \ref{fig:interval}).

Interval tree supports find and update operation. For update operation, usually not every leaf node is updated since the update information may stop at an intermediate node, recorded as lazy-tag, which means that it is temporarily suspended and will be executed in subsequent operations. Find operations need to be executed recursively starting from the root, and trigger the pass-down operation of all lazy-tags until the leaf node is reached. The complexity is $O(\log n)$ with respect to interval tree for both operations. We refer \cite{mulmuley1994computational}, Chapter 10.1 to readers for more detail.

\section{Algorithm}
\subsection{Overview}
In order to real-time display the voting status, the core of our algorithm is to record and maintain each node's ``lost voting power", initially valued zero. When a voter votes, his actual voting power is his total voting power minus his lost voting power. Meanwhile, some other voters' lost voting power should be updated after he votes. As long as each voter's ``lost voting power" can be updated within $O(\log n)$, the liquid democracy problem can be solved. 

See again the example in Figure \ref{fig:1}, 
\begin{itemize}
	\item After voter 1 votes for candidate $A$, all other voters' lost voting powers do not change. 
    \item After voter 5 votes for candidate $B$, $B$'s votes actually increase by $11-0=11$ (voter 5's total voting power minus lost voting power). Meanwhile, voter 2,3,4's lost voting powers all increase by 11. Lost voting powers of other voters {\em{that have not voted yet}} do not change  (lost voting powers for voters that  have voted is meaningless since each voter can only vote once).
    \item After voter 3 votes for candidate $C$, $C$'s votes actually increase by $33-11=22$ (voter 3's total voting power minus lost voting power). Meanwhile, voter 2's lost voting power increase by 22. Lost voting powers of other voters {\em{that have not voted yet}} do not change. 
\end{itemize}
It is not hard to get the following observation:
\begin{observation}
\label{obs:1}
	When a voter votes, a node's lost voting power needs to be updated if and only if the node lines on the path  from the voter to the voter's nearest parent node that has voted (we call it nearest voted parent for short, and voter 0 is regarded as voted. Note that in a tree, there is one and only one path from a node to one of his parent). 
\end{observation}
In the example, when voter 5 votes, the nearest voted parent is 1, so all nodes on the path $(5)\rightarrow4\rightarrow3\rightarrow2\rightarrow(1)$ are to be updated. When voter 3 votes, nodes on the path $(3)\rightarrow2\rightarrow(1)$ are to be updated. 

So the main two goals of our algorithm are:
\begin{enumerate}
	\item Find the voter's nearest voted parent.
	\item Update lost voting powers of nodes on the path from the voter to the voter's nearest voted parent.
\end{enumerate}
However, both 1 and 2 need to traverse the graph in traditional method, whose time complexity is $O(n)$ when the depth of the graph is high (say, chain like). Our solution is to use the interval tree to achieve $O(\log n)$, and use Merkel tree to for initialization.

The following two sequences are used, the preorder sequence and the bracket sequence. 
\begin{itemize}
	\item {\em Preorder sequence}, that is, traverse the tree in the order of $root\rightarrow leaf$ and record the nodes. In Figure \ref{fig:1}, the preorder sequence is 1 2 3 ... 12, that is, each node's index equals to its index of the preorder sequence. 
	\item {\em Bracket sequence}, that is, still traverse the tree in preorder but record the nodes when entering and exiting it respectively.  In Figure \ref{fig:1}, the bracket sequence is 1 2 3 4 5 6 6 5 4 7 8 8 7 3 2 9 10 10 11 11 12 12 9 1.
\end{itemize} 

Our algorithm consists of the following three parts.
\begin{itemize}
	\item In prepare period, each voter locally do an initialization of all voters' data, including their total voting powers, index of the preorder sequence, index of the bracket sequence and so on. The data are submit to the voting contract together with voters' voting massages, and are checked by the Merkel root (Section \ref{sec:step1}).
	\item For each voting massage, find the voter's nearest voted parent, based on the interval tree with respect to the preorder sequence (\ref{sec:step2}).
	\item For each voting  massage, update lost voting powers of nodes on the path from the voter to the voter's nearest voted parent, based on the interval tree with respect to the bracket sequence (\ref{sec:step3}).
\end{itemize}
\subsection{Notations}
We have the following variables:
\begin{itemize}
	\item $T$: The delegate tree, which is generated and stored off-chain. 
	\item $n$: Number of nodes, as well as the length of the preorder sequence (called preorder index for short). 
	\item $n_0$: Length of the bracket sequence. 
	\item $node$: A type, representing the voters. 
	\item $node.stake$: Node's voting power, which is given from the snapshot.
	\item $node.index$: Index of the node in the pre-order sequence.
	\item $node.address$: The Ethereum address of the node, which is an inherent  
	parameter.
	\item $b[]$: Mapping from a node's preorder index to the node.
	\item $nearestparent[]$: Mapping from a node's preorder index to its nearest voted parent's preorder index.
	\item $s[]$:  The score of the bracket sequence (showed in the following).
	\item $node.endpoint$: The maximum preorder index among the node's children (include multi-level).
	\item $node.left$: The first index where the node appears in the bracket sequence. 
	\item $node.right$: The second index where the node appears in the bracket sequence. 
	\item $node.power$: Node's total voting power (including its children's).
	\item $node.candidate$: The candidate that the voter votes.
	\item $v[]$: Recording the votes of candidates.
	\item $lazy_1[]$: Lazy-tag of the interval tree with respect to the preorder sequence, which also reflects the index of the nearest voted parent.
	\item $lazy_2[]$: Lazy-tag of the interval tree with respect to the bracket sequence, which also reflects the ``score" of the sequence.
\end{itemize}
All variables are global and initially valued 0 unless otherwise stated.  For other intermediate variables we will illustrate in the following subsection.
\subsection{Spare and Prepare Period}
\label{sec:step1}
\textbf{Spare Period}

A perpetual smart contract, called the {\em delegate contract}, are established. It has two methods:
\begin{itemize}
	\item $Delegate()$, voters call this method to appoint a delegate, or to undelegate (by delegate to an empty address). Before that, we recommend a protocol that each voter locally downloads the blockchain data and checks whether his delegate operation generates a cycle in the delegate graph. If so, the voter should change his delegate. The protocol can be integrated in the client. 
	\item $Vote()$, a holder calls this method to start a voting, and deploys a new smart contract, called the {\em voting contract}, which will be introduced in the following. After that, the prepare period begins. 
\end{itemize}

\textbf{Prepare Period}

a) At the beginning of the prepare period, all on-chain information are snapshotted by the current height of the blockchain, mainly, each voter's delegate and stake. Then, each voter involved in the voting locally constructs the delegate tree $T$ and gets all voters' voting powers according to the following rule. 
\begin{itemize}
	\item For each voter, get his last delegate operation from the snapshot and add the corresponding direct edge to the delegate graph. Then for all edges that are on a cycle, delete the latest edge. Then repeat the deletion until the delegate graph has no cycle\footnote{All transactions in Ethereum are attached with a time stamp. The time order on the blockchain is define that, if a transaction's block height is larger than another trasaction, then the former is later than the latter. If two transaction have the same block height, the transaction with the larger timestamp are later. The rule of Etherum guarantees that the timestamp can not be forged too far from the actual time otherwise the transaction is infeasible}. After that, add an edge from each zero-outdegree node to the virtual node, resulting in the delegate tree $T$. Since the rule is deterministic, all voters obtain a same delegate tree. 
	We show in next section  that the construction rule is incentive compatible for voters, 
	\item  Usually in a decentralized authority organization (DAO), a voter's voting power equals to one of his stake on the blockchain (say, one kind of ERC-20 token). So all voters' voting powers can be obtained from the snapshot ($node.stake$ in the notation). There are other off-chain method in practice to distribute voting powers, which is not critical of our paper as long as all voters can reach an agreement.
\end{itemize}

b) After the construction of $T$, we use $T.root$ to denote the root of $T$, i.e., the virtual node. Then, each voter locally call $Preorder(child)$, to obtain initialization data. ($n$ and $n_0$ is initialed $-1$ to exclude the virtual node).

\begin{algorithm}
	\label{alg:preorder}
	\textbf{Procedure} $Preorder(Node~root)$;
	\hrule
	$n \leftarrow n+1$\;
	$n_0 \leftarrow n_0+1$\;
	$root.left \leftarrow n_0$\;
	$root.index \leftarrow n$\;
	 $root.power \leftarrow root.stake$\;
	\For{$node$ in $root$'s direct child}
	{
		$Preorder(node)$
		$root.power \leftarrow root.power+node.power$\;
	}
	$root.endpoint \leftarrow n$\;
	$n_0 \leftarrow n_0+1$\;
	$root.right \leftarrow n_0$\;
\end{algorithm}
Intrinsically, a preorder traversal are executed, and all nodes' initialization data are computed simultaneously. 

c) Each voter construct the Merkel root according the initialization data, where the information of each leaf node is {\em hash(node.adddress, node.power, node.index, node.endpoint, node.leftbracket, node.rightbracket)}, representing the initialization data of a voter. The leaf nodes are ordered according to $node.index$ so that each voter's local Merkel tree are identical. The  Merkel root are hard-coded in the voting contract. (If the voting holder makes a mistake of the Merkel root, every voter can choose to ignore the voting contract and remind the holder to deploy a new contract)

After step (c), the prepare period is over and the voting period begins.
\subsection{Voting Period}
In this subsection, we describe the voting contract to process each direct voting massage. 

d) When a voter casts a direct vote, he sends a voting message which contains {\em(data, proof, node.candidate)}, where {\em data=(node.power, node.index, node.endpoint, node.leftbracket, node.rightbracket)} (here the node corresponds to the voter), the initialization data about himself. 

e) Upon receiving a voting massage {\em(data, proof, node.candidate)}, the voting contract first obtains the sender's Ethereum address, to check if it matches with	 $node.address$ in $data$.   If it matches, then the contract recovers a root according to $data$ and $proof$, and checks if the result matches to the Merkel root stored in the contract. If matches. the contract begins to process the voting massage, otherwise returns an ``error" response.

\begin{algorithm}
	\caption{Vote: upon receiving a voting message}
	\KwIn{$node$: voter}
	\KwIn{$data,proof,node.candidate$}
	\hrule
	
	\If {not check(RootHash, proof, data)}{
		\Return ;
	}
	$b[node.index]=node$\;
	$update2(node.left,node.left,1,2n,1,0)${\color{gray}
		//Find the value of node's leftbracket}\;
	$update2(node.right,node.right,1,2n,1,0)${\color{gray}
		//Find the value of node's rightbracket}\;
	$int~t=node.power-s[node.left]+s[node.right]$\;
	$C[node.candidate]+=t$\;
	$update1(node.index,node.index,1,n,1,0)${\color{gray}
		//Find the node's nearest parent}\;
	$Node~parent = b[nearestparent[node.index]]$\;
	$C[parent.candndate]-=t$\;
	$update1(node.index+1,node.endpoint,1,n,1,node.index)${\color{gray}
		//Update nearest voted parents}\;
	$update2(parent.left,node.left,1,2n,1,t)${\color{gray}
		\\//Update scores}\;
	Output $C[]$
	\label{alg:vote}
\end{algorithm}

f) The Algorithm \ref{alg:vote} shows the main procedure for processing a voting massage, which consists of the following instructions: 
\begin{itemize}
	\item Compute the voter's lost voting power.
	\item Define $t$ to be the voter's total voting power minus lost voting power, which represents his actual votes. 
	\item Find the voter's nearest voted parent, and then update other voters' nearest voterd parent (only the voter's children are affected).
	\item The ballot of the candidate that the voter's nearest parent votes decreases by $t$.
	\item Update the lost voting powers of the nodes on the path from the voter to the voter's nearest voted parent.
\end{itemize}
So far, the total procedure of algorithm is produced. In the next two subsections we will introduce the two functions $update1(),update2()$ 
\subsection{Find the Nearest Voted Parent}
\label{sec:step2}
The function $update1()$ achieves the goal of finding the voter's nearest voted parent, by implementing the update operation of the interval tree with respect to the preorder sequence. The lazy-tag of the interval tree's leaf node records the preorder index of corresponding voter's nearest voted parent. 

The following three observations are sufficient for the correctness:
\begin{observation}
\
	\begin{itemize}
	    \item A node's preorder index is always larger than its children's.
	    \item Preoroder indexes of a node's children are successive to the node's preorder index. 
	    \item When a voter votes, only his children's nearest voted parents need to be updated, which should be at least the voter's preorder index.
	\end{itemize}
\end{observation}
For a node (voter), indexes from $node.index+1$ to $node.endpoint$ represents the preorder indexes of its children, whose nearest voted parent need to be updated. Since they form a successive interval, the interval tree is applicable. 
\begin{algorithm}
	\textbf{procedure} $update1(int~L,int ~R, int~l, int~r, int~k, int~v)${\color{gray}
		\\//$[L,R]$ are the interval to be updated, $[l,r]$ is the current interval of the interval tree node, $k$ is the index of  interval tree node and $v$ is the value for updating.}
	\hrule
	\eIf {$L=l$ and $R=r$}
	{
		\If {$v>lazy_1[k]$} {$lazy_1[k] \leftarrow v$}
		\If {$L = R$} {$nearestparent[L] \leftarrow lazy_1[k]$}
		{\color{gray}
			//Recursion ends when updating interval equals to current node's interval, and then updating the value of the interval}
	}
	{
		$int~m \leftarrow (l+r)/2$\;
		\If {$lazy_1[2k]<lazy_1[k]$} {$lazy_1[2k] \leftarrow lazy_1[k]$}
		\If {$lazy_1[2k+1]<lazy_1[k]$} {$lazy_1[2k+1] \leftarrow lazy_1[k]${\color{gray}
				//pass down the lazy-tag}}
		\If {$L \leq m$}{$update1(L,\min\{m,R\},l,m,2k,v)$}
		\If {$R>m$}{$update1(\max\{m+1,L\},R,m+1,r,2k+1,v)$}
	}
\end{algorithm}

\begin{itemize}
	\item {\em update1(node.index+1, node.endpoint,1,n,1,node.index)} uses the voter's preorder index to do maximum-value update to (preorder indexes of) all its children's nearest voted parent (if the current value is less than the incoming parameter, then replace). 
	\item {\em update1(node.index,node.index,1,n,1,0)} finds and records the node's nearest voted parent (recorded in $nearestparent[]$). The value used for update is zero since no additional updating is needed (just trigger the pass-down operation of lazy-tags). Note that if there is no voted parent then 0 is record, the default value. 
\end{itemize}

\subsection{Update the Lost Voting Power}
\label{sec:step3}
When a voter votes, all nodes on the path from the voter to its nearest voted parented should update their lost voting powers. See Figure \ref{fig:2}, if voter 8 votes after voter 1 votes, then path $7\rightarrow3\rightarrow2\rightarrow1$ should be updated. However it is not a successive interval in the preorder sequence.
\begin{figure}
  \centering
  \begin{tikzpicture}
\Tree
[.\node[fsn](n1){1};
  [.\node[sn]{2};
   [.\node[sn]{3};
    [.\node[sn]{4};
     [.\node[sn]{5};
      [.\node[sn]{6};]]]
    [.\node[sn]{7};
     [.\node[fsn](n8){8};]]] ]
  [.\node[sn]{9};
   [.\node[sn]{10};]
   [.\node[sn]{11};]
   [.\node[sn]{12};]]
]
\node (a) at ($(n1.east) + (1, 0)$) {A};
\node (b) at ($(n8.east) + (1, 0)$) {B};
\draw [->, >=stealth', color=blue] (n1) -- (a);
\draw [->, >=stealth', color=blue] (n8) -- (b);

\end{tikzpicture}
	\caption{Updating a path}
	\label{fig:2}
\end{figure}
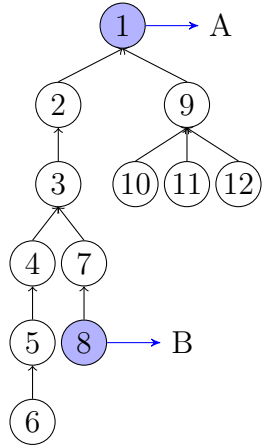

We use the bracket sequence to handle this problem. A bracket sequence is to
record each node twice in the pre-order traversal, one for enter and one for
exit, called left bracket and right bracket respectively. For Figure~\ref{fig:2}, the bracket sequence is
$$1,2,3,4,5,6,6,5,4,7,8,8,7,3,2,9,10,10,11,11,12,12,9,1$$

For a direct path starts from $u$ and ends at $v$, define the path's {\em bracket interval} to be indexed from $v.leftbracket$ to $u.leftbreacket$. The following observation shows the property of the bracket interval:
\begin{observation}
Given a direct path, for any node that does not lie on the path, it occur twice or does not occur in the path's bracket interval. For any node that lies on the path, it occur exactly once in the path's bracket interval. Moreover, only the node’s first appearance lies in the interval.
\end{observation}
For example, suppose the path from 8 to 1 is to be updated, its bracket interval is $1,2,3,4,5,6,6,5,4,7,8$ (index 1-11). Nodes 4,5,6,11,12 do not lie on the path, so they occur twice or do not occur in the interval. Nodes 1,2,3,7,8 lie in the path, so they occur once in the interval. 

We then define an array $s[]$ for recording the so-called ``score" of the bracket sequence. Given a path where the nodes' lost voting powers need to add some value, we add that value into the scores of the path's bracket interval. Then for a node's  lost voting power, we can compute it by
 $$s[node.left]-s[node.right]$$
 The reason is that, when we increase the score, only the nodes on the path increase their lost voting powers (only the leftbracket increases). For a node outside the path, the scores either does not change or both its leftbraket and rightbracket increase by the same value, thus the lost voting power does not change.

We construct another interval tree with respect to the bracket sequence  and maintain the score, recorded in the variable $lazy_2[]$ of leaf nodes. The function $update2()$ gives the implementation.
\begin{algorithm}
	\textbf{procedure} $update2(int~L,int ~R, int~l, int~r, int~k, int~v)$
	\hrule
	\eIf {$L=l$ and $R=r$}
	{
       $lazy_2[k] \leftarrow lazy_2[k]+v$\;
       \If {$L=R$}
       {
       	   $s[L]\leftarrow lazy_2[k]$
       }
	}
	{
		$int~m \leftarrow (l+r)/2$\;
		$lazy_2[2k] \leftarrow lazy_2[2k]+lazy_2[k]$\;
		$lazy_2[2k+1]\leftarrow lazy_2[2k+1]+lazy_2[k]$\;
		$lazy_2[k] \leftarrow 0$
			{\color{gray}
		//pass down the lazy-tag}\;
		\If {$L \leq m$}{$update2(L,\min\{m,R\},l,m,2k,v)$}
		\If {$R>m$}{$update2(\max\{m+1,L\},R,m+1,r,2k+1,v)$}
    }
\end{algorithm}
\begin{itemize}
	\item {\em update2(node.left,node.left,1,2n,1,0)} finds the score of the node's leftbracket (recorded in the array $s[]$). Similar for the rightbracket. 
	\item {\em update2(parent.left, node.left,1,2n,1,t)} update the score of the bracket interval of the path from the voter to its nearest parent.
\end{itemize}
So far our algorithm is introduced. In the next section we prove some theorems.
\section{Theorem}
In this section we prove some properties of our algorithm. We first analyze our protocol for constructing the delegate graph.
\begin{lemma}
   If a voter's delegate operation does not generate a cycle of the delegate graph (locally checked), then the corresponding edge will never be deleted. 
\end{lemma}
\begin{proof}
	Assume by contradiction that the delegate edge is deleted. By definition, there must be a cycle such that the delegate edge is the latest, which means that the cycle are generated by the appearance of the delegate edge, contradiction.
\end{proof}
It means that, if the voter follows the protocol, then his delegation are garanteed to be retained, which is benefit for him. Otherwise if he deviates (his delegation generates a cycle), his delegation may be deleted. (It is also possible to be retained, if other voters further change their delegate and remove the cycle)
\begin{lemma}
	If a voter deviates from our mechanism, by sending a delegate edge that generates a cycle, then further, this edge will not cause other voter's delegate edge to be refused if they follows the protocol.
\end{lemma}
\begin{proof}
	We call the delegate edge of the dishonest voter edge $A$. We prove that, if an edge $B$ is refused with the existence of edge $A$, then it will also be refused without the existence of edge $A$.
	
	Since $B$ is refused, it lies on a cycle which contains $A$. Since $A$ also lies on another cycle, if $A$ is delete, $B$ still lies on a cycle, and should be refused according to the protocol. The lemma is proved.
\end{proof}
The lemma shows that, even if a voter deviates from the protocol, other voters are not influenced. 

There are also sublinear-time algorithms that can judge whether a cycle is generated for a incoming delegate edge, which can be used in smart contract. However it is more complex and require more gas fee for each delegate massage. So still our protocol is recommended in practice. 

Next, we introduce our main theorem:
\begin{theorem}
	For each voting massage in liquid democracy problem, the voting status can be updated and displayed within $O(\log n)$ time complexity. Moreover, our algorithm can be deploy on the Ethereum mainnet and overcome the gas limitation, for the number of voters more than one million.
\end{theorem}
The theorem is obvious according to the properties of the tools we used. Here we illustrate some issues.
\begin{itemize}
	\item Processing a voter's voting message does not rely on the initialization data of other voters that has not voted, since our algorithm only requires the data from the nearest voted parent.
	\item The ``mapping" structure in Solidity (the coding language of Ethereum smart contract) satisfies that, the storages are allocated only if they are assigned values. For example, the storage $lazy[3]$ can be allocated without the allocation of $lazy[1]$ and $lazy[2]$. Moreover, $lazy[1]$ and $lazy[2]$ still can be visit but always returned a default value 0, which is just the requirement of our algorithm.
	\item The time complexity of update operation in interval tree is $O(\log n)$, since at each level of the interval tree, at most two intervals are at the recursive state: after the interval to be updated is first time separated to two sub-intervals, there are at least one endpoints of the interval to be update are the same as the endpoints of the interval of current node. So at the next level, either there are only one interval to be executed, or there are two intervals, but one of them is identical to the interval of the next interval tree node, and the recursion ends. 
\end{itemize}
For the part of Ethereum, we leave the proof in the experiment section.
\section{Experiment}
We compare between our algorithm and traversal algorithm (implemented by ourselves) by
recording the maximum consumed gas fee. Since in practice the only requirement is that the  consumed gas fee should be strictly limited by related Ethereum parameter, we focus on extreme cases that produce the maximum consumed gas fee. 

We assume that the delegate graph is a chain, which turns out to be the extreme case of consumed gas fee. We also consider extreme cases where the voter comes from the head/tail of the chain, which is possible to produce the maximum consumed gas fee for both algorithm.

\begin{figure}[h]
\begin{minipage}[t]{0.5\textwidth}
 \begin{tikzpicture}[scale=0.6]
\begin{axis}[votegas]
\addplot  plot coordinates {
(10, 139758)
(20, 210619)
(30, 281487)
(40, 352360)
(50, 423240)
(60, 494126)
(70, 565018)
(80, 635916)
(90, 706821)
(100, 777732)
(200, 1487185)
(300, 2197263)
(400, 2907966)
(500, 3619294)
(600, 4331247)
(700, 5043826)
(800, 5757029)
(900, 6470857)
};\addplot  plot coordinates {
(10, 520250)
(20, 580441)
(30, 559707)
(40, 640301)
(50, 619823)
(60, 619695)
(70, 700354)
(80, 700354)
(90, 700418)
(100, 679812)
(200, 739737)
(300, 820587)
(400, 799789)
(500, 799917)
(600, 880448)
(700, 880384)
(800, 859906)
(900, 859906)
(1000, 859970)
(2000, 919895)
(3000, 1000554)
};\legend{traversal,fast}
  \draw [dashed] ( axis cs:0, 6750000) -- ( axis cs:3200, 6750000) node
  [near start, above] {Gas Limit};
\end{axis}
\end{tikzpicture}
  \caption{Voting by the head.}
  \label{fig:eval:root}
\end{minipage}
\begin{minipage}[t]{0.5\textwidth}
 \begin{tikzpicture}[scale=0.6]
\begin{axis}[votegas]
\addplot  plot coordinates {
(10, 103692)
(20, 134473)
(30, 165254)
(40, 196036)
(50, 226818)
(60, 257601)
(70, 288384)
(80, 319167)
(90, 349951)
(100, 380735)
(200, 688599)
(300, 996502)
(400, 1304444)
(500, 1612425)
(600, 1920444)
(700, 2228503)
(800, 2536601)
(900, 2844739)
(1000, 3152915)
(2000, 6236825)
};\addplot  plot coordinates {
(10, 536968)
(20, 613237)
(30, 595379)
(40, 689505)
(50, 666551)
(60, 671584)
(70, 760741)
(80, 765774)
(90, 763257)
(100, 742819)
(200, 819088)
(300, 913534)
(400, 895548)
(500, 898193)
(600, 989610)
(700, 992255)
(800, 971817)
(900, 966912)
(1000, 974398)
(2000, 1050794)
(3000, 1144856)
};\legend{traversal,fast}
  \draw [dashed] ( axis cs:0, 6750000) -- ( axis cs:3200, 6750000) node
  [near start, above] {Gas Limit};
\end{axis}
\end{tikzpicture}
  \caption{Voting by the tail.}
  \label{fig:eval:tail}
\end{minipage}
\end{figure}

We conduct the evaluation on Ganache, which is a
personal blockchain for Ethereum development that can be used to deploy contracts, and run tests.
Our implementation can be found here~\footnote{\url{https://github.com/freeof123/liquid-voting/tree/master/ether-eval/contracts}}.

Our comparison is from two aspects. \begin{enumerate*}[1) ]%
  \item The voter is the \emph{head} of the delegate chain, as illustrated in
  Fig.~\ref{fig:eval:root}; \item the voter is the \emph{tail} of the delegate chain
, as shown in Fig.~\ref{fig:eval:tail}. \end{enumerate*}.
The gas limit is about 6,700,000 according to Ganache. Our
evaluation shows that \begin{itemize}
    \item the traversal algorithm performs better when the delegate chain
      is short, like smaller than 100;
    \item our algorithm significantly outperforms the traversal algorithm when
      the delegate chain is long enough;
    \item our algorithm can scale up with very limited gas increasing, while
      the traversal algorithm reaches the gas limit when the delegate chain
      grows up to 1,000.
\end{itemize}.

\section*{Acknowledgment}
We thank Ying Liu from Peking University for the advising the algorithm.

\bibliographystyle{plain}
\bibliography{liquid}

\begin{thebibliography}{10}

\bibitem{bannet2004hack}
Jonathan Bannet, David~W Price, Algis Rudys, Justin Singer, and Dan~S Wallach.
\newblock Hack-a-vote: Security issues with electronic voting systems.
\newblock {\em IEEE Security \& Privacy}, 2(1):32--37, 2004.

\bibitem{behrens2014principles}
Jan Behrens, Axel Kistner, Andreas Nitsche, and Bj{\"o}rn Swierczek.
\newblock {\em The principles of LiquidFeedback}.
\newblock Interaktive Demokratie e. V. Berlin, 2014.

\bibitem{birch2016full}
Sarah Birch.
\newblock Full participation: A comparative study of compulsory voting.
\newblock 2016.

\bibitem{blum2016liquid}
Christian Blum and Christina~Isabel Zuber.
\newblock Liquid democracy: Potentials, problems, and perspectives.
\newblock {\em Journal of Political Philosophy}, 24(2):162--182, 2016.

\bibitem{bocek2018smart}
Thomas Bocek and Burkhard Stiller.
\newblock Smart contracts--blockchains in the wings.
\newblock In {\em Digital Marketplaces Unleashed}, pages 169--184. Springer,
  2018.

\bibitem{brill2018pairwise}
Markus Brill and Nimrod Talmon.
\newblock Pairwise liquid democracy.
\newblock In {\em IJCAI}, pages 137--143, 2018.

\bibitem{carroll1884principles}
Lewis Carroll.
\newblock {\em The principles of parliamentary representation}.
\newblock Harrison and Sons, 1884.

\bibitem{christoff2017binary}
Zo{\'e} Christoff and Davide Grossi.
\newblock Binary voting with delegable proxy: An analysis of liquid democracy.
\newblock {\em arXiv preprint arXiv:1707.08741}, 2017.

\bibitem{gerber1998primary}
Elisabeth~R Gerber and Rebecca~B Morton.
\newblock Primary election systems and representation.
\newblock {\em Journal of Law, Economics, \& Organization}, pages 304--324,
  1998.

\bibitem{hanifatunnisa2017blockchain}
Rifa Hanifatunnisa and Budi Rahardjo.
\newblock Blockchain based e-voting recording system design.
\newblock In {\em 2017 11th International Conference on Telecommunication
  Systems Services and Applications (TSSA)}, pages 1--6. IEEE, 2017.

\bibitem{hardt2015google}
Steve Hardt and Lia~CR Lopes.
\newblock Google votes: a liquid democracy experiment on a corporate social
  network.
\newblock 2015.

\bibitem{harris2004black}
Bev Harris and David Allen.
\newblock {\em Black box voting: Ballot tampering in the 21st century}.
\newblock Talion Pub., 2004.

\bibitem{hjalmarsson2018blockchain}
Friorik~P Hjalmarsson, Gunnlaugur~K Hreiarsson, Mohammad Hamdaqa, and Gisli
  Hjalmtysson.
\newblock Blockchain-based e-voting system.
\newblock In {\em 2018 IEEE 11th International Conference on Cloud Computing
  (CLOUD)}, pages 983--986. IEEE, 2018.

\bibitem{kahng2018liquid}
Anson Kahng, Simon Mackenzie, and Ariel~D Procaccia.
\newblock Liquid democracy: An algorithmic perspective.
\newblock In {\em Thirty-Second AAAI Conference on Artificial Intelligence},
  2018.

\bibitem{kiayias2002self}
Aggelos Kiayias and Moti Yung.
\newblock Self-tallying elections and perfect ballot secrecy.
\newblock In {\em International Workshop on Public Key Cryptography}, pages
  141--158. Springer, 2002.

\bibitem{kohno2004analysis}
Tadayoshi Kohno, Adam Stubblefield, Aviel~D Rubin, and Dan~S Wallach.
\newblock Analysis of an electronic voting system.
\newblock In {\em IEEE Symposium on Security and Privacy, 2004. Proceedings.
  2004}, pages 27--40. IEEE, 2004.

\bibitem{kousser2007does}
Thad Kousser and Megan Mullin.
\newblock Does voting by mail increase participation? using matching to analyze
  a natural experiment.
\newblock {\em Political Analysis}, 15(4):428--445, 2007.

\bibitem{mccorry2017smart}
Patrick McCorry, Siamak~F Shahandashti, and Feng Hao.
\newblock A smart contract for boardroom voting with maximum voter privacy.
\newblock In {\em International Conference on Financial Cryptography and Data
  Security}, pages 357--375. Springer, 2017.

\bibitem{miller1969program}
James~C Miller.
\newblock A program for direct and proxy voting in the legislative process.
\newblock {\em Public choice}, 7(1):107--113, 1969.

\bibitem{mulmuley1994computational}
Ketan Mulmuley.
\newblock Computational geometry.
\newblock {\em An Introduction Through Randomized Algorithms. Prentice-Hall},
  1994.

\bibitem{nakamoto2008bitcoin}
Satoshi Nakamoto et~al.
\newblock Bitcoin: A peer-to-peer electronic cash system.
\newblock 2008.

\bibitem{paulin2014through}
Alois Paulin.
\newblock Through liquid democracy to sustainable non-bureaucratic government.
\newblock In {\em Proc. Int. Conf. for E-Democracy and Open Government}, pages
  205--217, 2014.

\bibitem{schulze2011new}
Markus Schulze.
\newblock A new monotonic, clone-independent, reversal symmetric, and
  condorcet-consistent single-winner election method.
\newblock {\em Social Choice and Welfare}, 36(2):267--303, 2011.

\bibitem{william1912primary}
William~S. U'Ren.
\newblock Government by proxy now.
\newblock {\em New York Times}, 1912.

\bibitem{wood2014ethereum}
Gavin Wood et~al.
\newblock Ethereum: A secure decentralised generalised transaction ledger.
\newblock {\em Ethereum project yellow paper}, 151(2014):1--32, 2014.

\bibitem{yang2018decentralized}
Xuechao Yang, Xun Yi, Surya Nepal, and Fengling Han.
\newblock Decentralized voting: a self-tallying voting system using a smart
  contract on the ethereum blockchain.
\newblock In {\em International Conference on Web Information Systems
  Engineering}, pages 18--35. Springer, 2018.

\end{thebibliography}
\end{document}